\documentclass[a4paper,twoside,10pt]{article}
\usepackage[USenglish]{babel} %francais, polish, spanish, ...
\usepackage[T1]{fontenc}
\usepackage[ansinew]{inputenc}

\usepackage{lmodern} %Type1-font for non-english texts and characters

\usepackage{graphicx} %%For loading graphic files
\usepackage{amsmath}
\usepackage{amsthm}
\usepackage{amsfonts}
\usepackage{url}

\newtheorem{theorem}{Theorem}
\newtheorem*{theorem*}{Theorem}

\newtheorem*{example*}{Example}

\newtheorem*{lemma*}{Lemma}
\newtheorem{corollary}[theorem]{Corollary}
\newtheorem*{corollary*}{Corollary}

\begin{document}

\pagestyle{empty} %No headings for the first pages.

%% Title Page %%%%%%%%%%%%%%%%%%%%%%%%%%%%%%%%%%%%%%%%%%%%%%%
%% ==> Write your text here or include other files.

%% The simple version:
\title{An Improvement on Rank of Explicit Tensors}
\author{Benjamin Weitz}
%\date{} %%If commented, the current date is used.
\maketitle
\begin{abstract}
We give constructions of $n^k \times n^k \times n$ tensors of rank at least $2n^k - O(n^{k-1})$. As a corollary we obtain an $[n]^r$ shaped tensor with rank at least $2n^{\lfloor r/2\rfloor} - O(n^{\lfloor r/2 \rfloor - 1})$ when $r$ is odd. The tensors are constructed from a simple recursive pattern, and the lower bounds are proven using a partitioning theorem developed by Brockett and Dobkin. These two bounds are improvements over the previous best-known explicit tensors that had ranks $n^k$ and $n^{\lfloor r/2\rfloor}$ respectively. \end{abstract}
%% The nice version:
%\input{titlepage} %%You need a file 'titlepage.tex' for this.
%% ==> TeXnicCenter supplies a possible titlepage file
%% ==> with its templates (File | New from Template...).

%% Inhaltsverzeichnis %%%%%%%%%%%%%%%%%%%%%%%%%%%%%%%%%%%%%%%
\clearpage
\pagestyle{plain} %Now display headings: headings / fancy / ...

\section{Introduction}
An important and well-studied property of linear operators, equivalently matrices, is their rank. Much is understood about the rank of matrices over fields, and an efficient algorithm exists for the calculation of the rank of an explicit matrix. However, a closely related problem, calculating the rank of a tensor, a generalized version of a matrix, has been shown to be {\bf NP}-complete\cite{Has90}, and so is unlikely to have an efficient algorithm. Due to the intractability of the problem, very few results have been shown on this subject.
\subsection{Importance of Tensor Rank}
The rank of a tensor is relevant and important in several different settings. Fast matrix multiplication, a problem that is of incredible importance, can be improved by improving the upper bound on the rank of a related tensor. A recent paper by Ran Raz proved two theorems relating lower bounds on the rank of tensors and lower bounds on the size of arithmetic formulas:
\begin{itemize}
\item {\bf Theorem}: {\it Let $A: [n]^r \rightarrow \mathbb{F}$ be a tensor such that $r\leq O(\log{n}/\log{\log{n}})$. If there exists a formula of size $n^c$ for the polynomial 
\[f_A(x_{1,1}, \dots, x_{r,n}) = \sum_{i_1,\dots i_r\in[n]} A(i_1,\dots, i_r) \cdot \prod_{j=1}^r x_{j,i_j}\]
then the tensor rank of $A$ is at most $n^{r\cdot(1-2^{-O(c)})}$}\cite{Raz10}. \\
\item {\bf Corollary}: {\it Let $A: [n]^r$ be a tensor such that $r\leq O(\log{n}/\log{\log{n}})$. If the tensor rank of $A$ is $\geq n^{r\cdot(1-o(1))}$ then there is no polynomial size formula for the polynomial $f_A$} \cite{Raz10}.\end{itemize}These two theorems give a strong motivation behind finding explicit $[n]^r$ tensors of high rank. In this paper we give an explicit hypercube tensor with rank approaching $2n^{\lfloor r/2\rfloor}$, an improvement over the previous best-known example by a constant factor of $2$. 
\subsection{Methodology}
For each integer $k$, we will give an $n^k \times n^k \times n$ tensor with rank at least $2n^k - O(n^{k-1})$, an improvement over the previous best-known $n^k$. To do so, we will use a partitioning theorem developed by Brockett and Dobkin in \cite{Bro76}. This theorem allows us to lower bound the rank of a tensor that is formed by concatenating or gluing together other tensors, provided they are sufficiently different. We will construct a tensor recursively by continually gluing together three copies of a smaller tensor. The partitioning theorem will allow us to lower bound the rank of the tensor at each step, and thus the final tensor as well. As a corollary we will construct an $\underbrace{n \times \dots \times n}_{r \text{ times}}$ tensor of rank $2n^{\lfloor r/2\rfloor} - O(n^{\lfloor r/2\rfloor - 1})$ when $r$ is odd by viewing the first construction under an isomorphism. This is an improvement over the previous best-known  $n^{\lfloor r/2\rfloor}$ by a constant factor.
 
\subsection{Definitions}
Throughout this paper, $F$ will denote a field. Let $A \in F^{n_1\times n_2\times n_3}$. If 
\[A = x_1 \otimes x_2 \otimes x_3\]
for $x_i \in F^{n_i}$ and $A$ is nonzero, then $A$ is called a \emph{simple} or rank-$1$ tensor. The rank of a general tensor $A$ is defined as the minimal number $r$ such that we can write
\begin{equation}
A = \sum_{i=1}^r B_i
\end{equation}
where each $B_i$ is a simple tensor. This is a natural extension of matrix rank, because if $x_3 = 1$ then the rank of $A$ agrees with the matrix rank. Throughout this paper, $R[A]$ will denote the rank of $A$. 
\subsubsection{Slices, Concatenations, and the Characteristic Matrix}
Let $A \in F^{n_1\times n_2\times n_3}$, fix a positive integer $1 \leq k \leq n_3$ and let $B \in F^{n_1 \times n_2 \times 1}$ satisfy 
\[B_{ij} = A_{ijk}\]
Then $B$ is called the $k$th slice of $A$. We will denote the $k$th slice of a tensor $A$ as $A_k$. The \emph{concatenation} of tensors $A \in F^{n_1\times n_2\times m}$ and $B \in F^{n_1\times n_2\times m'}$, denoted $AB \in F^{n_1\times n_2\times (m+m')}$, is the tensor such that
\[AB_{ijk} = \left\{\begin{tabular}{ll} $A_{ijk}$ & if $1 \leq k \leq m$ \\ $B_{ij(k-m)}$ & if $m < k \leq m+m'$ \end{tabular}\right.\]
The $m+m'$ slices of the concatenation are the $m$ slices of $A$ followed by the $m'$ slices of $B$. Also $AB$ and $BA$ differ only by permutations of the indices in the third dimension, so $R[AB] = R[BA]$. 
\\
\\
The \emph{characteristic matrix} of $A \in F^{n_1\times n_2\times n_3}$ is a matrix $A(s)$ with indeterminants, i.e. $A(s) \in E^{n_1\times n_2}$, where $E = F \cup S$ and $S = \{s_i\}_{i=1}^{n_3}$ is a set of indeterminates with
\[A(s) = \sum_{i=1}^{n_3} s_iA_i\]
so each indeterminate represents the values on a different slice. Define $\dim s = |S| = n_3$. Define the \emph{column} (resp. \emph{row}) \emph{rank} to be the maximal number of linearly independent columns (resp. rows) as in \cite{Bro76}. Note that the row and column rank are not necessarily equal, for example $A(s) = \left[s_1\hspace{4mm}s_2\right]$ has column rank$=2$ and row rank$=1$. We also sometimes write $R[A(s)]$ for $R[A]$. To avoid trivialities, we usually work with \emph{nondegenerate} tensors; a tensor $A$ is nondegenerate if no nontrivial linear combination of its slices vanish and its characteristic matrix $A(s)$ has full row and column rank. An analogy of concatenation can be defined with characteristic matrices as well. Let $A(s)$ and $B(t)$ be two characteristic matrices of the same dimensions, and let $\dim s = n$ and $\dim t = m$. We define $C(u) = A(s) + B(t)$ as 
\[C(u) = \sum_{i=1}^n A_is_i + \sum_{j=1}^m B_jt_j\]
where $u = s \cup t$. Note $C(u)$ is the characteristic matrix of $AB$, so this addition can be considered as a concatenation.

\subsection{The Partitioning Theorem}
The main tool we use in our construction is the partitioning theorem developed by Brockett and Dobkin in \cite{Bro76}, and we write it here for easy referral:
\begin{theorem}\label{brdob1}
Let $G(s)$ be a nondegenerate characteristic matrix, and let one of the following cases hold:
\begin{enumerate}
\item[(i)] $G(s) = \left[\begin{tabular}{c} $G_1(s)$ \\ $G_2(s)$ \end{tabular}\right]$ \\
\item[(ii)] $G(s) = \left[\begin{tabular}{cc} $G_1(s)$ & $G_2(s)$\end{tabular}\right]$ \\
\item[(iii)]  $G(s) = G_1(u) + G_2(v)$
\end{enumerate}
Then for each case we have
\begin{enumerate}
\item[(i)] $R[G(s)] \geq \min_M R[G_1(s) + NG_2(s)] + \text{row rank }G_2(s)$ \\
\item[(ii)] $R[G(s)] \geq \min_N R[G_1(s) + G_2(s)M] + \text{column rank }G_2(s)$ \\
\item[(iii)] $R[G(s)] \geq \min_T R[G_1(u) + G_2(Tu)] + \dim v$
\end{enumerate}
for matrices $M$, $N$, and $T$ sized so that the two summands are the same shape and the addition is well-defined, and juxtaposition means regular matrix multiplication. 
\end{theorem}
This theorem essentially states that if two halves of a tensor "don't overlap too much", then each slice of the second half must add at least one to the rank. A special case of "don't overlap too much" is given in the following theorem:

\begin{theorem}\label{brdob2}
Let $G_1(s)$, $G_2(s)$, and $G_3(s)$ all be nondegenerate characteristic matrices. Then we have
\begin{enumerate}
\item[(i)] $R\left[\begin{tabular}{cc} $G_1(s)$ & $0$ \\ $G_2(s)$ & $G_3(s)$ \end{tabular}\right] \geq \max\{R[G_1(s)] + \text{column rank }G_3(s), R[G_3(s)] + \text{row rank }G_1(s)\}$ \\
\item[(ii)] $R\left[\begin{tabular}{cc} $G_1(s) + G_2(t)$ & $G_3(t)$\end{tabular}\right]\geq \max\{R[G_1(s)] + \text{column rank }G_3(t), R[G_3(t)] + \dim s\}$ \\
\item[(iii)] $R\left[\begin{tabular}{c} $G_1(s) + G_2(t)$ \\ $G_3(t)$\end{tabular}\right] \geq \max\{R[G_1(s)] + \text{row rank }G_3(t), R[G_3(t)] + \dim s\}$
\end{enumerate}
\end{theorem}

\section{The Main Result}
In this section we give a construction that yields $n^k \times n^k \times n$ tensors of rank approaching $2n^k$. These numbers are, to the best of our knowledge, the largest known rank of any explicit tensor of these shapes. As a corollary, for $r$ odd, these constructions allow us to give an $[n]^r$ shaped tensor of rank approaching $2n^{\lfloor r/2 \rfloor}$, another improvement to the best of our knowledge. The first step is to prove a lower bound for a block tensor:
\begin{theorem}\label{main} 
Let $A \in F^{m \times n \times p}$ be nondegenerate, $B \in F^{m \times n' \times p'}$ be nondegenerate, and $C \in F^{m' \times n \times p'}$ be nondegenerate and let $E \in F^{m\times n \times p'}$, and let ${\bf 0}$ be the tensor of zeroes of appropriate dimensions to be concatenated, and let
\[M = \left[\begin{tabular}{cc} $AE$ & ${\bf 0}B$ \\ ${\bf 0}C$ & ${\bf 00}$\end{tabular}\right]\]
then
\[R[M] \geq R[A] + \text{column rank }B(t) + \text{row rank }C(t)\]\end{theorem}
\begin{proof}
First, transforming into characteristic matrices,
\[M(u) = \left[\begin{tabular}{cc} $A(s) + E(t)$ & $B(t)$ \\ $C(t)$ & $0$ \end{tabular}\right]\] 
with $u = s\cup t$.
We partition 
\[M(u) = \left[\begin{tabular}{c} $G_1(u)$ \\ $G_2(u)$ \end{tabular}\right]\] 
with 
\begin{align*}
G_1(u) &= \left[\begin{tabular}{cc} $A(s)+E(t)$ & $B(t)$\end{tabular}\right] \\
G_2(u) &= \left[\begin{tabular}{cc} $C(t)$ & $0$\end{tabular}\right]\end{align*}  
By Theorem \ref{brdob1}, 
\[R[M(u)] \geq \min_N R[G_1(u) + NG_2(u)] + \text{row rank }G_2(u)\]
By Theorem \ref{brdob2}, 
\begin{align*}
R[G_1(u) + NG_2(u)] &= R\left[A(s)+(E+NC)(t) \hspace{4mm} B(t)\right] \\
&\geq \max\{R[A(s)] + \text{column rank }B(t), R[B(t)] + \dim s\}\end{align*} 
Since $\text{row rank }G_2(u) = \text{row rank }C(s)$, we have 
\[R[M(u)] \geq \text{row rank }C(s) + R[A(s)] + \text{column rank }B(t).\]
\end{proof}
This theorem is the key to our construction. We recursively build a tensor as follows: pick a positive integer $k$ and let $A^{(0)} = I_{n^{k-1}}$, and define
\[A^{(i+1)} = \left[\begin{tabular}{cc} $A^{(i)}{\bf 0}$ & ${\bf 0}A^{(i)}$ \\ ${\bf 0}A^{(i)}$ & ${\bf 00}$\end{tabular}\right]\]
the main result is 
\begin{theorem}\label{nknkn}
Pick $l = \log n$. Then the tensor $A^{(l)}$ above has dimensions $n^k \times n^k \times n$ and satisfies $R[A^{(l)}] \geq 2n^k - O(n^{k-1})$. 
\end{theorem}
\begin{proof}
For any $i$, it is clear that $A^{(i)}$ is a $2^in^{k-1} \times 2^in^{k-1} \times 2^i$ tensor. Furthermore, an easy induction shows that $A^{(i)}(s)$ is nondegenerate by noting that $A^{(i)}$ always has at least one slice with full row and column rank, and a nontrivial linear combination of slices of $A^{(i)}$ that vanish is such a combination of slices of $A^{(i-1)}$ as well. Thus 
\[\text{row rank }A^{(i)}(s) = \text{column rank }A^{(i)}(s) = 2^in^{k-1}\] 
and $A^{(i)}$ is nondegenerate. By Theorem \ref{main},
\[R[A^{(i+1)}] \geq R[A^{(i)}] + \text{row rank }A^{(i)}(s) + \text{column rank }A^{(i)}(s) = R[A^{(i)}] + 2^{i+1}n^{k-1}\]
Then a straightforward induction shows 
\begin{align*}
R[A^{(i)}] &\geq R[A^{(0)}] + \sum_{j=0}^{i-1} 2^{j+1}n^{k-1} \\
&= n^{k-1} + 2(2^i - 1)n^{k-1}
\end{align*}
setting $l = \log n$, we have $R[A^{(l)}] \geq 2n^k - n^{k-1}$ and $A_i \in F^{n^k\times n^k\times n}$.\end{proof}
This construction allows us to improve on the previous best-known explicit hypercube tensor by taking the preimage of these tensors under the canonical isomorphism. 
\begin{corollary}\label{improvnr}
Let $r$ be odd, $k = \lfloor r/2\rfloor$, $A^{(l)}$ as above, and let $\phi$ be the canonical isomorphism
\begin{align*}
\phi: F^{\overbrace{n \times \dots \times n}^{r \text{ times}}} &\rightarrow F^{n^k \times n^k \times n} \\
x_1 \otimes \dots \otimes x_r &\mapsto (x_1 \otimes \dots \otimes x_k) \otimes (x_{k+1} \otimes \dots \otimes x_{2k}) \otimes x_r
\end{align*}
then $\phi^{-1}(A^{(l)})$ is an $\underbrace{n \times \dots \times n}_{r \text{ times}}$ tensor with rank at least $2n^{\lfloor r/2\rfloor} - O(n^{\lfloor r/2\rfloor - 1})$. 
\end{corollary}
\begin{proof}
We show that for any tensor $B \in F^{\overbrace{n\times \dots \times n}^{r \text{ times}}}$, $R[\phi(B)] \leq R[B]$. Assume the opposite towards a contradiction. Then if 
\[B = \sum_{i=1}^{R[B]} D_i\]
for simple tensors $D_i$, we have
\[\phi(B) = \sum_{i=1}^{R[B]} \phi(D_i)\]
and as $D_i$ is simple, so is $\phi(D_i)$, but since $R[\phi(B)] > R[B]$ this contradicts minimality of $R[\phi(B)]$, thus $R[\phi(B)] \leq R[B]$, so clearly $R[A^{(l)}] \leq R[\phi^{-1}(A^{(l)})]$.
\end{proof} 
To our knowledge, these are the best-known ranks for explicit $[n]^r$ and $n^k \times n^k \times n$ tensors for any $k$, including the important cube $n \times n \times n$ tensors. 

\section{Conclusion}
In this paper we have presented an improvement to about $2n^k$ from the previous highest rank explicit tensors for the $n^k \times n^k \times n$ shape. This extends to an improvement for the shape $[n]^r$ when $r$ is odd. These tensors were constructed by using Brockett and Dobkin's partitioning theorem in a recursive manner. However, using this theorem imposes a restriction on the quality of the lower bounds. In order to improve further, we need to either improve the partitioning theorem or develop a different method.
\subsection{Open Problems}
\begin{itemize}
\item The most important open problem is the one presented as the motivation for this paper. The improvements in this paper do not come anywhere close to the $n^{r(1-o(1))}$ threshold for hypercube tensors. Any explicit tensor with this rank would imply super-polynomial lower bounds on certain functions as per Ran Raz's recent theorem\cite{Raz10}. Any attempt to develop examples of high-rank tensors should keep this goal in mind. \\
\item An improvement to Brockett and Dobkin's partitioning theorem would be extremely useful. The same techniques presented here would be more powerful and perhaps improve by an increase in the exponent, rather than a constant factor. 
\end{itemize}
\subsection{Additional Notes}
This paper is the result of research done at Caltech from June 2010 to August 2010 as part of the SURF program. I worked under Chris Umans, Professor of Computer Science, and I'd like to thank him for all his help and advice while working on this project. Additionally, in between the writing and the submission of this article, an independent article was published by Boris Alexeev, Michael Forbes, and Jacob Tsimerman\cite{Aft11} that gives, among other things, an explicit $n^k \times n^k \times n$ $\{0,1\}$-tensor with rank at least $2n^k + n - \Theta(k\log n)$. The techniques in this paper are similar to those described here, so the two bounds are very close, but the one given by Alexeev, Forbes, and Tsimerman has better lower-order terms. Interested parties can read the paper here \url{http://arxiv.org/abs/1102.0072}.

\end{document}